\documentclass[sigconf,natbib=true]{acmart}

\usepackage{textcomp}
\usepackage{stfloats}
\usepackage{url}
\usepackage{verbatim}
\usepackage{amsthm}
\usepackage{subfigure}
\usepackage{makecell}
\usepackage{multirow}
\usepackage{enumitem}
\usepackage{subcaption}

\usepackage{color, colortbl}

\definecolor{Gray}{rgb}{0.9, 0.9, 0.9}

\usepackage[capitalize,noabbrev]{cleveref}

\newtheorem{definition}{Definition}

\newtheorem{proposition}{Proposition}

\usepackage[ruled,linesnumbered]{algorithm2e}

\usepackage[inkscapelatex=false]{svg}

\copyrightyear{2026}
\acmYear{2026}
\setcopyright{cc}
\setcctype{by}
\acmConference[CIKM '26]{Proceedings of the 35th ACM International Conference on Information and Knowledge Management}{November 07--11, 2026}{Rome, Italy}
\acmBooktitle{Proceedings of the 35th ACM International Conference on Information and Knowledge Management (CIKM '26), November 07--11, 2026, Rome, Italy}
\acmDOI{10.1145/3799682.3839882}
\acmISBN{979-8-4007-2539-5/2026/11}

\begin{document}

\title[Neural Tree Collaborative Filtering]{Neural Tree Collaborative Filtering: Rethinking Graph Collaborative Filtering as Tree Collaborative Filtering with Curvature-Aware Propagation Depth}

\author{Jinfeng Xu}
\email{jinfeng@connect.hku.hk}
\affiliation{%
  \institution{The University of Hong Kong}
  \city{HongKong SAR}
  \country{China}}

\author{Zheyu Chen}
\email{zheyu.chen@bit.edu.cn}
\affiliation{%
  \institution{Beijing Institute of Technology}
  \city{Beijing}
  \country{China}}

\author{Ziyue Peng}
\email{pengziyue001@gmail.com}
\affiliation{%
  \institution{The Hong Kong University of Science and Technology}
  \city{HongKong SAR}
  \country{China}}

\author{Shuo Yang}
\email{shuoyang.ee@gmail.com}
\affiliation{%
  \institution{The University of Hong Kong}
  \city{HongKong SAR}
  \country{China}}

\author{Jinze Li}
\email{lijinze-hku@connect.hku.hk}
\affiliation{%
  \institution{The University of Hong Kong}
  \city{HongKong SAR}
  \country{China}}

\author{Wenhao Yuan}
\email{wenhao.yuan@connect.hku.hk}
\affiliation{%
  \institution{The University of Hong Kong}
  \city{HongKong SAR}
  \country{China}}

\author{Jian Chen}
\email{ccccccj03@connect.hku.hk}
\affiliation{%
  \institution{The University of Hong Kong}
  \city{HongKong SAR}
  \country{China}}

\author{Shujie Li}
\email{shujie.li@connect.hku.hk}
\affiliation{%
  \institution{The University of Hong Kong}
  \city{HongKong SAR}
  \country{China}}

\author{Edith C. H. Ngai}
\authornote{Corresponding author.}
\email{chngai@eee.hku.hk}
\affiliation{%
  \institution{The University of Hong Kong}
  \city{HongKong SAR}
  \country{China}}

\renewcommand{\shortauthors}{Jinfeng Xu et al.}

 
\begin{abstract}
Graph Collaborative Filtering (GCF) has become the dominant paradigm in modern recommender systems by modeling user--item interactions as a bipartite graph and propagating embeddings through a fixed number of message-passing layers. However, applying a uniform propagation depth to every node ignores a fundamental property of real interaction graphs: nodes differ substantially in their local connectivity, so peripheral nodes quickly suffer from over-smoothing while hub-like nodes remain under-explored beyond their immediate neighborhood. In this paper, we revisit GCF from a tree-structured perspective and propose Neural Tree Collaborative Filtering (NTCF), a framework that re-interprets each node's local neighborhood as a rooted tree and assigns a node-specific propagation depth based on a closed-form local-degree-imbalance score that serves as a discrete Ricci-curvature proxy. We provide a theoretical analysis showing that (i) NTCF strictly generalizes NGCF, degenerating to NGCF when all curvature-induced depth adjustments vanish (a lower bound on its representation power), and (ii) the curvature-aware schedule retains strictly more discriminative information at deep layers on positively-curved (peripheral) nodes than uniform-depth propagation. NTCF can achieve higher performance than most widely used GCF backbone models and can be integrated into existing advanced self-supervised models as a backbone, replacing their original backbone to achieve enhanced performance. Extensive experiments on three public datasets demonstrate the superiority of NTCF\footnote[1]{Code is available at: \href{https://github.com/Jinfeng-Xu/NTCF}{https://github.com/Jinfeng-Xu/NTCF}.}.
\end{abstract}



\maketitle

\section{Introduction}
Recommender systems are widely used to alleviate information overload on the Web~\cite{xu2025cohesion,he2017neural,xu2024aligngroup,xu2024mentor}. Collaborative Filtering (CF) addresses the recommendation problem by learning user preferences from historical interactions of similar users. Recently, Graph Collaborative Filtering (GCF)~\cite{wang2019neural,xu2026survey,chen2025squeeze,he2020lightgcn,liu2021interest,mao2021ultragcn,chen2020revisiting} has become the dominant paradigm because the user-item interaction data is naturally graph-structured. Representative models such as NGCF~\cite{wang2019neural} adopt vanilla GCN architectures with feature transformation and nonlinear activation, while LightGCN~\cite{he2020lightgcn} simplifies these components and has since served as the backbone of many follow-up GCF and self-supervised models~\cite{yu2022graph,xu2025enhancing,chen2025hypercomplex,zhang2024recdcl,lin2022improving}.

Despite its empirical success, the prevailing GCF paradigm shares a structural limitation that has received little attention: \emph{every node propagates information through exactly the same number of layers $L$}. From the perspective of each individual node, message passing can be viewed as constructing a rooted tree of depth $L$ in which the node itself is the root and its (multi-hop) neighbors form the descendants. Forcing this tree to have the same depth for all nodes is, in fact, unnatural. Real user-item interaction graphs are highly heterogeneous: peripheral nodes whose neighbors are dense local hubs see information that is already heavily mixed at one or two hops away, so deepening their tree merely re-aggregates over-smoothed signals~\cite{chen2020revisiting,zhao2020pairnorm,liu2020towards}; in contrast, locally-hub-like nodes whose immediate neighbors are themselves comparatively peripheral need \emph{deeper} exploration to reach the informative collaborative signals that lie beyond their immediate neighborhood.

This observation prompts a question that lies at the heart of GCF design: \emph{should the propagation depth be a node-level property rather than a global hyper-parameter?} If propagation can be made node-adaptive while preserving the bidirectional nature of the user-item graph, GCF can in principle be re-interpreted as an ensemble of node-rooted trees, with each tree's depth tailored to the local geometry of the graph at its root.

Motivated by this perspective, we propose \textbf{Neural Tree Collaborative Filtering (NTCF)}. NTCF re-interprets the bipartite interaction graph as a collection of rooted trees, one per node, in which message passing fills the tree from the root outwards. Crucially, the depth of each tree is no longer a global hyper-parameter but is determined by a closed-form, $O(|\mathcal{E}|)$-computable \emph{local degree-imbalance score} that serves as a discrete-Ricci-curvature proxy~\cite{ollivier2009ricci,lin2011ricci}. The score provides a principled, data-driven measure of how a node's local connectivity compares to that of its neighbors, and translates directly into a node-specific depth: negative-curvature local hubs receive a deeper tree to reach far collaborative signals, positive-curvature peripheral nodes receive a shallower tree to avoid re-aggregating already-mixed signals, and curvature-flat nodes keep the standard depth. Because this scheme reduces to standard NGCF propagation when no node receives a curvature adjustment, NTCF provably contains NGCF as a special case while admitting a broader family of node-adaptive propagation patterns.


\section{Preliminary}
\subsection{NGCF Brief}
Let $\mathcal{U}$ and $\mathcal{I}$ denote the user and item sets, respectively, and let $\mathcal{G}=(\mathcal{V},\mathcal{E})$ be the bipartite interaction graph with $\mathcal{V}=\mathcal{U}\cup\mathcal{I}$ and $\mathcal{E}\subseteq\mathcal{U}\times\mathcal{I}$. Each node $v\in\mathcal{V}$ has an embedding $\mathbf{e}_v^{(0)}\in\mathbb{R}^d$. NGCF~\cite{wang2019neural} performs $L$ layers of bipartite message passing; for any node $v\in\mathcal{V}$ with neighbor set $\mathcal{N}_v$,
\begin{equation}
\small
\mathbf{e}_v^{(l+1)} = \sigma\!\Big(\mathbf{W}_1\mathbf{e}_v^{(l)} + \!\!\sum_{u\in\mathcal{N}_v}\!\!\frac{\mathbf{W}_1\mathbf{e}_u^{(l)} + \mathbf{W}_2(\mathbf{e}_u^{(l)}\!\odot\mathbf{e}_v^{(l)})}{\sqrt{|\mathcal{N}_v||\mathcal{N}_u|}}\Big),
\label{eq:ngcf}
\end{equation}
where $\mathbf{W}_1,\mathbf{W}_2$ are trainable weights shared across users and items, $\sigma$ is a nonlinear activation, and $\odot$ is the element-wise product. Final embeddings are obtained by concatenating the layer-wise outputs, $\hat{\mathbf{e}}_v=\mathbf{e}_v^{(0)}\|\mathbf{e}_v^{(1)}\|\cdots\|\mathbf{e}_v^{(L)}$.

\subsection{Tree-Structured View of GCF}
A useful but underexplored fact about Eq.~\ref{eq:ngcf} is that, from the standpoint of any single node $v$, $L$-layer GCN propagation produces an embedding that depends \emph{only} on the rooted tree obtained by unrolling $\mathcal{G}$ around $v$ to depth $L$. We formalize this tree as follows.
\begin{definition}[Rooted Propagation Tree]
For a node $v\in\mathcal{V}$ and depth $L\in\mathbb{N}$, the rooted propagation tree $\mathcal{T}_v^{(L)}$ is defined by breadth-first unrolling of $\mathcal{G}$ from $v$ for $L$ hops. Each node in $\mathcal{T}_v^{(L)}$ corresponds to a walk of length at most $L$ originating at $v$, and edges follow the bipartite structure of $\mathcal{G}$.
\end{definition}
Under this view, GCN-style propagation can be equivalently described as: for every $v\in\mathcal{V}$, build $\mathcal{T}_v^{(L)}$ and aggregate information from leaves to root. This makes explicit the fact that NGCF and its successors implicitly assume \emph{the same tree depth $L$ for every node}---a global hyper-parameter that ignores the heterogeneous local geometry of $\mathcal{G}$. We relax this assumption by allowing each node to have its own depth $L_v$, and we use a closed-form local-curvature score (Sec.~\ref{sec:curvature}) to determine $L_v$.

\subsection{Local Degree-Imbalance Curvature on Graphs}
\label{sec:curvature}
Discrete-graph analogues of Ricci curvature, such as Ollivier-Ricci curvature~\cite{ollivier2009ricci} and the Lin-Lu-Yau~\cite{lin2011ricci} variant, characterize how a node's local connectivity compares to that of its neighbors. Exact computation of these quantities, however, requires solving an optimal-transport problem at every edge, which is impractical for industrial-scale interaction graphs containing millions of edges. We therefore adopt a simple, closed-form curvature \emph{proxy} that captures the same local-imbalance intuition while admitting an $O(|\mathcal{E}|)$ pre-computation. Given a node $v$ with degree $d(v)$, we define its local degree-imbalance score $\kappa(v)$ as:
\begin{equation}
\kappa(v) \;=\; \frac{\bar{d}_{\mathcal{N}(v)} - d(v)}{\bar{d}_{\mathcal{N}(v)} + d(v) + \epsilon},
\label{eq:curvature}
\end{equation}
where $\bar{d}_{\mathcal{N}(v)}=\frac{1}{|\mathcal{N}(v)|}\sum_{w\in\mathcal{N}(v)}d(w)$ is the average degree of $v$'s neighbors and $\epsilon>0$ avoids division by zero. The interpretation is intuitive. When $d(v)>\bar{d}_{\mathcal{N}(v)}$, the node is itself more connected than its neighbors and $\kappa(v)<0$: $v$ acts as a local hub whose neighbors are comparatively peripheral, so a deeper propagation tree can reach informative two- or three-hop signals that a shallow tree would miss. When $d(v)<\bar{d}_{\mathcal{N}(v)}$, the node sits at the periphery of a more densely connected local cluster and $\kappa(v)>0$: deepening the tree merely re-aggregates information that is already heavily mixed in the local hub, accelerating over-smoothing, so a shallower tree is preferable. When $\kappa(v)\!\approx\!0$, the node and its neighbors are comparably connected and the standard depth is appropriate. Eq.~\ref{eq:curvature} is bounded in $[-1,1]$ and is computed once at preprocessing time. For brevity, we simply refer to $\kappa(v)$ as the \emph{curvature} of $v$ in the rest of the paper.

\section{Methodology}
We present Neural Tree Collaborative Filtering (NTCF), illustrated in \cref{fig:framework}. 

\begin{figure}[t]
\centering
\includegraphics[width=1\linewidth]{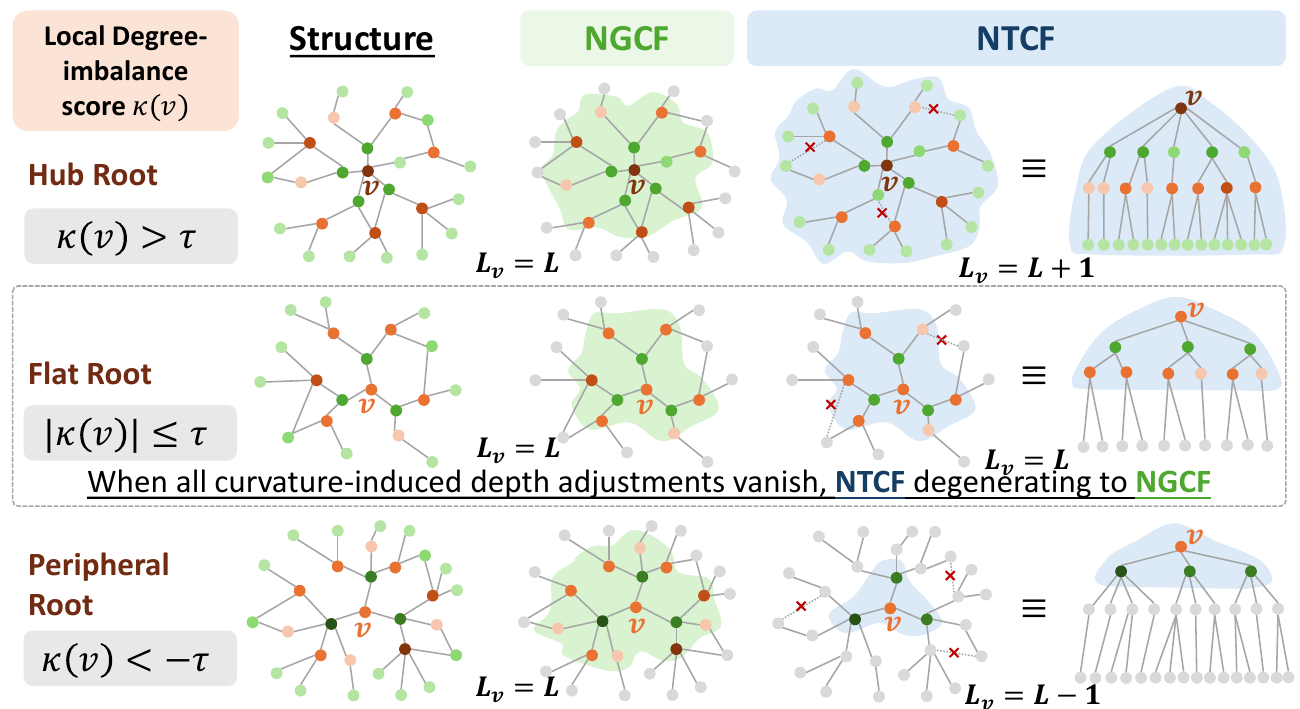}
\vskip -0.1in
\caption{Overview of NTCF.}
\label{fig:framework}
\vskip -0.15in
\end{figure}

\subsection{Curvature-Aware Propagation Depth}
\label{sec:depth}
Given a global base depth $L$ and a curvature threshold $\tau\!\geq\!0$, NTCF assigns each node $v$ an effective depth:
\begin{equation}
L_v \;=\; \max\!\big(L + \Delta(v),\,1\big), \qquad
\Delta(v) =
\begin{cases}
+1, & \kappa(v) < -\tau,\\
\;\;0, & |\kappa(v)| \leq \tau,\\
-1, & \kappa(v) > +\tau,
\end{cases}
\label{eq:depth}
\end{equation}
so that $L_v\!\in\!\{\max(L\!-\!1,1),\,L,\,L\!+\!1\}$. The threshold $\tau$ sets a curvature-flat band around zero that controls the fraction of adjusted nodes, and the clip guarantees every node at least one round of message passing. Eq.~\ref{eq:depth} is computed once at preprocessing in $O(|\mathcal{E}|)$ and cached.

\subsection{Tree-Structured Propagation}
\label{sec:propagation}
Let $\mathbf{E}^{(0)}\!\in\!\mathbb{R}^{(|\mathcal{U}|+|\mathcal{I}|)\times d}$ stack the layer-0 embeddings of all nodes, and let $L_{\max}=\max_v L_v$. NTCF performs $L_{\max}$ message-passing steps using a bipartite GNN convolution $f_\theta(\cdot)$, where $f_\theta(\cdot)$ realizes Eq.~\ref{eq:ngcf} via efficient sparse matrix multiplication and shares parameters across users and items. After each layer, NTCF applies a curvature-driven node mask that decides, per node, whether the new layer output is kept or whether the node has already exhausted its depth budget and should retain its layer-$0$ representation. Concretely, define the binary mask $\mathbf{m}^{(l)}\in\{0,1\}^{|\mathcal{V}|}$ with $\mathbf{m}_v^{(l)}=\mathbb{I}[L_v>l]$, and the per-layer update:
\begin{equation}
\widetilde{\mathbf{E}}^{(l+1)} \;=\; \mathrm{Norm}(\mathrm{Drop}(\sigma_a(f_\theta({\mathbf{E}}^{(l)})))),
\label{eq:propagate}
\end{equation}
\begin{equation}
\mathbf{E}^{(l+1)} \;=\; \mathbf{m}^{(l)} \!\odot\! \widetilde{\mathbf{E}}^{(l+1)} \;+\; (\mathbf{1}-\mathbf{m}^{(l)}) \!\odot\! \mathbf{E}^{(0)},
\label{eq:mask}
\end{equation}
where $\sigma_a(\cdot)$ denotes LeakyReLU, $\mathrm{Drop}(\cdot)$ message dropout, and $\mathrm{Norm}(\cdot)$ $\ell_2$ row normalization. Eq.~\ref{eq:mask} keeps updating a node $v$ while $l<L_v$ and otherwise resets it to its initial embedding $\mathbf{e}_v^{(0)}$. Crucially, masking applies to a node's \emph{own} state only: a masked-out $v$ still feeds its layer-$0$ value to still-active neighbors $w\in\mathcal{N}(v)$. This realizes the ``one tree per root'' interpretation---the rooted tree at $v$ terminates at depth $L_v$ while the bidirectional graph is preserved throughout.

\paragraph{Final Aggregation.}
After $L_{\max}$ layers, NTCF aggregates the per-layer embeddings into a final node representation. We support two aggregation strategies. The first, used by default, is layer concatenation, consistent with NGCF: $\hat{\mathbf{E}} = \mathbf{E}^{(0)}\,\|\,\mathbf{E}^{(1)}\,\|\,\cdots\,\|\,\mathbf{E}^{(L_{\max})}$. The second is curvature-weighted mean aggregation, which averages each node's embedding only over its effective layers:
\begin{equation}
\hat{\mathbf{e}}_v \;=\; \frac{1}{L_v + 1}\sum_{l=0}^{L_v}\mathbf{e}_v^{(l)}.
\label{eq:agg}
\end{equation}
The user and item representations $\hat{\mathbf{E}}_{\mathcal{U}}, \hat{\mathbf{E}}_{\mathcal{I}}$ are obtained by splitting $\hat{\mathbf{E}}$. The score for a user-item pair $(u,i)$ is calculated by $\hat{\mathbf{e}}_u^\top \hat{\mathbf{e}}_i$.

\paragraph{Optimization.}
We optimize NTCF with the Bayesian Personalized Ranking (BPR) loss~\cite{rendle2009bpr}:
\begin{equation}
\mathcal{L} = \!\!\sum_{(u,i_p,i_n)\in\mathcal{D}}\!\! -\ln \sigma_s\!\big(\hat{\mathbf{e}}_u^\top \hat{\mathbf{e}}_{i_p} - \hat{\mathbf{e}}_u^\top \hat{\mathbf{e}}_{i_n}\big) + \lambda\|\Theta\|^2,
\label{eq:bpr}
\end{equation}
where $\sigma_s$ is the sigmoid function (distinct from $\sigma_a$ in Eq.~\ref{eq:propagate}), $\mathcal{D}$ is the training set of triplets, and $\lambda$ controls $\ell_2$ regularization.

\paragraph{Complexity.}
The curvature pre-computation in Eq.~\ref{eq:depth} costs $O(|\mathcal{E}|)$ and is performed once and cached. Each training epoch shares the same per-layer sparse-aggregation cost as NGCF, $O(L_{\max}\!\cdot\!|\mathcal{E}|\!\cdot\!d)$, plus a negligible $O(|\mathcal{V}|\!\cdot\!d)$ masking step per layer. Since $L_{\max}\!\leq\!L\!+\!1$, NTCF's training cost exceeds NGCF's by at most a single message-passing layer, with no recurring overhead from the cached curvature.

\subsection{Theoretical Analysis}
\label{sec:theory}
We establish two results: NTCF contains NGCF as a special case (\cref{prop:lower}, a lower bound on representation power), and NTCF preserves a strictly positive layer-$L$ Dirichlet-energy floor that NGCF, under the standard over-smoothing contraction assumption, provably loses (\cref{prop:smoothing}).

\begin{proposition}[NGCF as the Zero-Adjustment Limit]
\label{prop:lower}
Let $\mathrm{NTCF}_\tau$ denote NTCF with curvature threshold $\tau$, and let $\mathrm{NGCF}_L$ denote NGCF with $L$ layers. If $\tau \geq \tau^\star := \max_{v\in\mathcal{V}}|\kappa(v)|$, then for every node $v$ we have $\Delta(v)=0$, hence $L_v=L$ and the masks $\mathbf{m}^{(l)}$ are identically one. The propagation in Eq.~\ref{eq:propagate}--Eq.~\ref{eq:mask} then reduces to $L$ applications of the bipartite convolution $f_\theta$, which by construction realizes the NGCF update of Eq.~\ref{eq:ngcf}. With $\mathrm{Norm}$ and $\mathrm{Drop}$ disabled and concatenation aggregation, $\mathrm{NTCF}_\tau$ recovers exactly the NGCF$_L$ propagation and final-embedding construction; with them enabled, $\mathrm{NTCF}_\tau$ realizes a normalized, dropout-regularized variant whose hypothesis space contains that of $\mathrm{NGCF}_L$.
\end{proposition}
\vspace{-0.05in}
\begin{proof}[Proof]
For $\tau\geq\tau^\star$, $|\kappa(v)|\leq\tau$ holds for all $v$, so the middle branch of Eq.~\ref{eq:depth} is taken everywhere, $\Delta(v)\!=\!0$, and $\mathbf{m}_v^{(l)}\!=\!\mathbb{I}[L_v\!>\!l]\!=\!1$ for all $l<L$. Substituting into Eq.~\ref{eq:mask} gives $\mathbf{E}^{(l+1)}\!=\!\widetilde{\mathbf{E}}^{(l+1)}$ at every layer, i.e.\ the masking has no effect and each layer applies $f_\theta$, the NGCF update of Eq.~\ref{eq:ngcf}. Concatenating $\{\mathbf{E}^{(l)}\}_{l=0}^{L}$ matches NGCF's final embedding construction. When $\mathrm{Norm}$ and $\mathrm{Drop}$ are disabled this is exactly $\mathrm{NGCF}_L$; otherwise the two extra operations only enlarge the realizable hypothesis space.
\end{proof}
\vspace{-0.05in}
Thus a properly tuned NTCF is never worse than NGCF.

\begin{proposition}[Anti-Smoothing]
\label{prop:smoothing}
Let $\mathrm{D}_l(\mathbf{E}) := \sum_{(u,v)\in\mathcal{E}}\|\mathbf{e}_u^{(l)}\!-\!\mathbf{e}_v^{(l)}\|^2$ denote the layer-$l$ Dirichlet energy, and let $\mathbf{P}_\perp$ denote the orthogonal projection onto the \emph{energy-relevant subspace} (the complement of the dominant eigenvector of the normalized adjacency). Following the standard over-smoothing analysis~\cite{oono2020graph,liu2020towards,zhao2020pairnorm}, assume the per-layer propagation operator restricted to this subspace is $\rho$-Lipschitz with $\rho<1$, and let $R_0\!:=\!\max_v\|\mathbf{P}_\perp\mathbf{e}_v^{(0)}\|$ bound the energy-relevant component of the initial embeddings. Let $\mathcal{V}_{+}\!=\!\{v\!:\!\kappa(v)\!>\!\tau\}$ be the positive-curvature nodes; each receives a $-1$ depth adjustment ($L_v\!=\!L\!-\!1$), so its layer-$L$ state is reset to $\mathbf{e}_v^{(0)}$ in NTCF. Writing $\delta_v\!:=\!\|\mathbf{P}_\perp\mathbf{e}_v^{(0)}\|$, we have
\begin{equation}
\resizebox{0.92\hsize}{!}{$
\!\!\!\mathrm{D}_L(\mathbf{E}_{\textsc{ngcf}}) \!\leq\! \rho^{2L}\!\cdot\!\mathrm{D}_0(\mathbf{E}^{(0)}),
\;\;
\mathrm{D}_L(\mathbf{E}_{\textsc{ntcf}}) \!\geq\! \sum_{v\in\mathcal{V}_+}\!\sum_{w\in\mathcal{N}(v)}\!\!\big(\delta_{v}\!-\!\rho^L\!R_0\big)_{+}^{2}\!,$}
\label{eq:smoothbound}
\end{equation}
where $(\cdot)_+\!=\!\max(\cdot,0)$. As $L\!\to\!\infty$, NGCF's upper bound decays as $\rho^{2L}\!\to\!0$, while NTCF's lower bound increases to $\sum_{v\in\mathcal{V}_+}\!\deg(v)\,\delta_{v}^{2}$. Hence, provided at least one reset node retains a non-degenerate energy-relevant component (i.e.\ $\delta_v\!>\!0$ for some $v\!\in\!\mathcal{V}_+$, which holds almost surely under generic random initialization), for sufficiently large $L$ we have $\mathrm{D}_L(\mathbf{E}_{\textsc{ntcf}})\!>\!\mathrm{D}_L(\mathbf{E}_{\textsc{ngcf}})$: NTCF strictly preserves more discriminative information at deep layers.
\end{proposition}
\vspace{-0.05in}
\begin{proof}[Proof]
\textbf{NGCF upper bound.} The standard analysis~\cite{oono2020graph} bounds the Dirichlet energy by the energy of the projected differences, which contract under the $\rho$-Lipschitz assumption: $\|\mathbf{P}_\perp(\mathbf{e}_u^{(L)}\!-\!\mathbf{e}_v^{(L)})\|\!\leq\!\rho^L\|\mathbf{P}_\perp(\mathbf{e}_u^{(0)}\!-\!\mathbf{e}_v^{(0)})\|$. Squaring and summing over $\mathcal{E}$ yields $\mathrm{D}_L(\mathbf{E}_{\textsc{ngcf}})\!\leq\!\rho^{2L}\mathrm{D}_0(\mathbf{E}^{(0)})$.
\textbf{NTCF lower bound.} For $v\!\in\!\mathcal{V}_+$ the mask sets $\mathbf{e}_v^{(L)}\!=\!\mathbf{e}_v^{(0)}$ exactly (since $\mathbf{m}_v^{(L-1)}\!=\!\mathbb{I}[L_v\!>\!L\!-\!1]\!=\!0$ as $L_v\!=\!L\!-\!1$), so its energy-relevant component is preserved: $\|\mathbf{P}_\perp\mathbf{e}_v^{(L)}\|\!=\!\delta_v$. For any neighbor $w$ that propagates normally, the contraction shrinks its energy-relevant component toward zero: $\|\mathbf{P}_\perp\mathbf{e}_w^{(L)}\|\!\leq\!\rho^L\|\mathbf{P}_\perp\mathbf{e}_w^{(0)}\|\!\leq\!\rho^L R_0$. Since the orthogonal projection $\mathbf{P}_\perp$ is non-expansive, applying it and then the reverse triangle inequality gives
$\|\mathbf{e}_v^{(L)}\!-\!\mathbf{e}_w^{(L)}\|\!\geq\!\|\mathbf{P}_\perp(\mathbf{e}_v^{(L)}\!-\!\mathbf{e}_w^{(L)})\|\!\geq\!\big|\,\|\mathbf{P}_\perp\mathbf{e}_v^{(L)}\|\!-\!\|\mathbf{P}_\perp\mathbf{e}_w^{(L)}\|\,\big|\!\geq\!(\delta_v\!-\!\rho^L R_0)_+$.
Squaring and summing over all edges incident to $\mathcal{V}_+$ yields the stated lower bound.
\end{proof}
\vspace{-0.05in}
\cref{prop:smoothing} formalizes the intuition behind Eq.~\ref{eq:depth}: retaining the layer-$0$ representations of otherwise over-smoothed peripheral nodes sustains a Dirichlet-energy floor as $L$ grows. Together with \cref{prop:lower}, this makes NTCF strictly better than NGCF in deep regimes whenever $|\mathcal{V}_+|>0$, and no worse otherwise.

\begin{table*}[!t]
\centering
\caption{Performance comparison of baselines and NTCF on the three datasets in terms of R@$K$ and N@$K$. Best results are in \textbf{bold}.}
\vskip -0.15in
\label{tab:main}
\small
\resizebox{\linewidth}{!}{
\begin{tabular}{l|cccccc|cccccc|cccccc}
\toprule
Datasets & \multicolumn{6}{c|}{Kindle} & \multicolumn{6}{c|}{Pinterest} & \multicolumn{6}{c}{Yelp} \\
Metrics & R@10 & R@20 & R@50 & N@10 & N@20 & N@50 & R@10 & R@20 & R@50 & N@10 & N@20 & N@50 & R@10 & R@20 & R@50 & N@10 & N@20 & N@50 \\
\midrule
NCF~\cite{he2017neural}   & 0.1185 & 0.1583 & 0.2232 & 0.0744 & 0.0850 & 0.0989 & 0.0777 & 0.1276 & 0.2365 & 0.0478 & 0.0621 & 0.0868 & 0.0531 & 0.0885 & 0.1654 & 0.0377 & 0.0486 & 0.0685 \\
NGCF~\cite{wang2019neural}    & 0.1306 & 0.1770 & 0.2340 & 0.0796 & 0.0919 & 0.1086 & 0.0870 & 0.1428 & 0.2633 & 0.0545 & 0.0721 & 0.1018 & 0.0630 & 0.1026 & 0.1864 & 0.0446 & 0.0567 & 0.0784 \\
LightGCN~\cite{he2020lightgcn}  & 0.1570 & 0.2080 & 0.2888 & 0.0981 & 0.1117 & 0.1290 & 0.1000 & 0.1621 & 0.2862 & 0.0635 & 0.0830 & 0.1136 & 0.0730 & 0.1163 & 0.2016 & 0.0520 & 0.0652 & 0.0875 \\
UltraGCN~\cite{mao2021ultragcn}   & 0.1520 & 0.2013 & 0.2791 & 0.0939 & 0.1082 & 0.1241 & 0.0967 & 0.1588 & 0.2817 & 0.0613 & 0.0808 & 0.1098 & 0.0758 & 0.1180 & 0.2029 & 0.0543 & 0.0666 & 0.0891 \\
IMP-GCN~\cite{liu2021interest}    & 0.1551 & 0.2057 & 0.2850 & 0.0963 & 0.1101 & 0.1269 & 0.0985 & 0.1603 & 0.2845 & 0.0624 & 0.0814 & 0.1113 & 0.0751 & 0.1182 & 0.2017 & 0.0539 & 0.0662 & 0.0885 \\
FKAN-GCF~\cite{xu2024fourierkan}    & 0.1580 & 0.2089 & 0.2900 & 0.0985 & 0.1125 & 0.1290 & 0.1003 & 0.1614 & 0.2860 & 0.0633 & 0.0827 & 0.1129 & 0.0755 & 0.1183 & 0.2019 & 0.0540 & 0.0666 & 0.0895 \\
\midrule
\textbf{NTCF (Ours)} & \textbf{0.1600} & \textbf{0.2104} & \textbf{0.2913} & \textbf{0.0999} & \textbf{0.1137} & \textbf{0.1309} & \textbf{0.1023} & \textbf{0.1635} & \textbf{0.2888} & \textbf{0.0642} & \textbf{0.0838} & \textbf{0.1141} & \textbf{0.0769} & \textbf{0.1200} & \textbf{0.2036} & \textbf{0.0547} & \textbf{0.0669} & \textbf{0.0905} \\
\midrule
\textbf{w/o curvature} & 0.1306 & 0.1770 & 0.2340 & 0.0796 & 0.0919 & 0.1086 & 0.0870 & 0.1428 & 0.2633 & 0.0545 & 0.0721 & 0.1018 & 0.0630 & 0.1026 & 0.1864 & 0.0446 & 0.0567 & 0.0784 \\
\textbf{r/ mean agg.} & 0.1589 & 0.2090 & 0.2902 & 0.0993 & 0.1133 & 0.1296 & 0.1015 & 0.1626 & 0.2875 & 0.0636 & 0.0832 & 0.1135 & 0.0761 & 0.1187 & 0.2019 & 0.0541 & 0.0660 & 0.0898 \\
\bottomrule
\end{tabular}
}
\vskip -0.1in
\end{table*}

\begin{table*}[!t]
\centering
\caption{Compatibility analysis: replacing the LightGCN backbone of advanced self-supervised models with NTCF.}
\vskip -0.15in
\label{tab:ssl}
\small
\resizebox{\linewidth}{!}{
\begin{tabular}{ll|cccccc|cccccc|cccccc}
\toprule
\multicolumn{2}{c|}{Datasets} & \multicolumn{6}{c|}{Kindle} & \multicolumn{6}{c|}{Pinterest} & \multicolumn{6}{c}{Yelp} \\
Method & Backbone & R@10 & R@20 & R@50 & N@10 & N@20 & N@50 & R@10 & R@20 & R@50 & N@10 & N@20 & N@50 & R@10 & R@20 & R@50 & N@10 & N@20 & N@50 \\
\midrule
\multirow{2}{*}{SimGCL~\cite{yu2022graph}}     
& LightGCN & 0.1659 & 0.2138 & 0.2903 & 0.1052 & 0.1184 & 0.1350 & 0.1051 & 0.1576 & 0.2442 & 0.0705 & 0.0871 & 0.1086 & 0.0908 & 0.1331 & 0.2130 & 0.0673 & 0.0811 & 0.1039\\
& \textbf{NTCF}    & \textbf{0.1698} & \textbf{0.2180} & \textbf{0.2941} & \textbf{0.1077} & \textbf{0.1205} & \textbf{0.1372} & \textbf{0.1084} & \textbf{0.1609} & \textbf{0.2475} & \textbf{0.0717} & \textbf{0.0884} & \textbf{0.1101} & \textbf{0.0924} & \textbf{0.1352} & \textbf{0.2156} & \textbf{0.0689} & \textbf{0.0830} & \textbf{0.1056}\\
\midrule
\multirow{2}{*}{LightGCL~\cite{cailightgcl}}   
& LightGCN & 0.1601 & 0.2098 & 0.2901 & 0.0996 & 0.1133 & 0.1303 & 0.0881 & 0.1322 & 0.2383 & 0.0534 & 0.0673 & 0.0981 & 0.0766 & 0.1188 & 0.2008 & 0.0551 & 0.0681 & 0.0899\\
& \textbf{NTCF}    & \textbf{0.1642} & \textbf{0.2136} & \textbf{0.2930} & \textbf{0.1013} & \textbf{0.1149} & \textbf{0.1320} & \textbf{0.0910} & \textbf{0.1356} & \textbf{0.2407} & \textbf{0.0552} & \textbf{0.0691} & \textbf{0.0997} & \textbf{0.0778} & \textbf{0.1205} & \textbf{0.2021} & \textbf{0.0558} & \textbf{0.0689} & \textbf{0.0911}\\
\midrule
\multirow{2}{*}{RecDCL~\cite{zhang2024recdcl}} 
& LightGCN & 0.1613 & 0.2108 & 0.2861 & 0.1015 & 0.1142 & 0.1314 & 0.1021 & 0.1619 & 0.2898 & 0.0663 & 0.0839 & 0.1144 & 0.0850 & 0.1304 & 0.2115 & 0.0609 & 0.0755 & 0.0989\\
& \textbf{NTCF}    & \textbf{0.1640} & \textbf{0.2133} & \textbf{0.2888} & \textbf{0.1041} & \textbf{0.1159} & \textbf{0.1334} & \textbf{0.1036} & \textbf{0.1642} & \textbf{0.2916} & \textbf{0.0678} & \textbf{0.0855} & \textbf{0.1163} & \textbf{0.0869} & \textbf{0.1322} & \textbf{0.2139} & \textbf{0.0622} & \textbf{0.0770} & \textbf{0.1000}\\
\midrule
\multirow{2}{*}{NLGCL~\cite{xu2025nlgcl}} 
& LightGCN & 0.1702 & 0.2160 & 0.2945 & 0.1081 & 0.1208 & 0.1379 & 0.1148 & 0.1793 & 0.3089 & 0.0760 & 0.0948 & 0.1267 & 0.0952 & 0.1414 & 0.2327 & 0.0703 & 0.0853 & 0.1094\\
& \textbf{NTCF}    & \textbf{0.1720} & \textbf{0.2181} & \textbf{0.2964} & \textbf{0.1093} & \textbf{0.1224} & \textbf{0.1397} & \textbf{0.1170} & \textbf{0.1816} & \textbf{0.3109} & \textbf{0.0773} & \textbf{0.0965} & \textbf{0.1288} & \textbf{0.0968} & \textbf{0.1433} & \textbf{0.2348} & \textbf{0.0716} & \textbf{0.0873} & \textbf{0.1117}\\
\bottomrule
\end{tabular}
}
\vskip -0.1in
\end{table*}

\begin{table}[!t]
\centering
\caption{Statistics of all datasets.}
\vskip -0.15in
\label{tab:dataset}
\setlength{\tabcolsep}{6pt}
\small
\begin{tabular}{lcccc}
\toprule
Dataset & \#Users & \#Items & \#Interactions & Sparsity \\
\midrule
Kindle    & 60,468 & 57,212 & 880,859 & 99.975\% \\
Pinterest & 55,188 & 9,912 & 1,445,622 & 99.736\% \\
Yelp      & 45,477 & 30,708 & 1,777,765 & 99.873\% \\
\bottomrule
\end{tabular}
\vskip -0.15in
\end{table}

\section{Experiment}

\subsection{Datasets and Evaluation Metrics}
To validate NTCF, we experiment on three public datasets: \textbf{Yelp}~\cite{xu2025nlgcl}, \textbf{Amazon Kindle store}~\cite{he2016ups}, and \textbf{Pinterest}~\cite{geng2015learning}. They vary in scale and sparsity. For Yelp we use a 15-core setting (minimum 15 interactions), and for Kindle and Pinterest we filter out users and items with fewer than 5 interactions. Following prior works~\cite{lin2022improving,zhang2024recdcl,xu2024fourierkan,chen2025squeeze}, we adopt a user-based split with training, validation, and test sets in an $8\!:\!1\!:\!1$ ratio. We report two established metrics~\cite{he2020lightgcn}: Recall@$K$ (R@$K$) and NDCG@$K$ (N@$K$) for $K\in\{10,20,50\}$. All results are averaged over $5$ runs with different seeds.


\subsection{Baselines and Experimental Settings}
To verify the effectiveness of our NTCF, we compare it with six widely used GCN-based backbone models: \textbf{NCF}~\cite{he2017neural}, \textbf{NGCF}~\cite{wang2019neural}, \textbf{LightGCN}~\cite{he2020lightgcn}, \textbf{UltraGCN}~\cite{mao2021ultragcn}, \textbf{IMP-GCN}~\cite{liu2021interest}, and \textbf{FKAN-GCF}~\cite{xu2024fourierkan}. Moreover, to verify NTCF's compatibility with self-supervised methods, we replace the original LightGCN backbone of four advanced self-supervised models with NTCF: \textbf{SimGCL}~\cite{yu2022graph}, \textbf{LightGCL}~\cite{cailightgcl}, \textbf{RecDCL}~\cite{zhang2024recdcl}, and \textbf{NLGCL}~\cite{xu2025nlgcl,xu2026nlgclp}. For a fair comparison, we fix the embedding size to $64$ for all models, initialize embeddings with Xavier initialization~\cite{glorot2010understanding}, and use Adam~\cite{kingma2014adam} as the optimizer. Hyper-parameters of all baselines are tuned following their published papers. For NTCF, we fix $\lambda=10^{-4}$, tune the layer number $L$ from $\{1,2,3,4\}$, the curvature threshold $\tau$ from $\{0.05, 0.1, 0.2, 0.4\}$, and the message dropout ratio from $\{0.0, 0.1, 0.2, 0.3\}$.

\subsection{Performance Comparison}
The results are reported in \cref{tab:main}. NTCF consistently outperforms all baselines on both metrics across all three datasets, which we attribute to curvature-aware propagation: peripheral nodes are spared late-layer over-smoothing, while hub nodes absorb farther collaborative signals. Its only overhead over NGCF is the per-node curvature score, computed once in $O(|\mathcal{E}|)$ and cached.

\subsection{Compatibility Analysis}
We further evaluate the compatibility of NTCF with advanced self-supervised graph-based recommenders by replacing their default LightGCN backbones with NTCF. As shown in \cref{tab:ssl}, all four self-supervised methods (SimGCL, LightGCL, RecDCL, and NLGCL) consistently benefit from NTCF on every dataset and metric. This confirms that adapting propagation depth to the local geometry of each node is a property combinable with self-supervised contrastive objectives, and that NTCF can be deployed as a drop-in backbone improvement for the modern SSL-based recommendation pipeline.

\subsection{Ablation Study}
\cref{tab:main} also reports two ablations. \textbf{w/o curvature} sets $\tau\!\to\!\infty$ so that $\Delta(v)\!\equiv\!0$ for all nodes, recovering the NGCF-equivalent of NTCF predicted by \cref{prop:lower}. \textbf{r/ mean agg.} replaces concatenation with the curvature-weighted mean aggregation in Eq.~\ref{eq:agg}. The curvature-aware variant outperforms the uniform-depth variant on all datasets, empirically corroborating the strict-improvement direction predicted by \cref{prop:smoothing}, and concatenation slightly outperforms mean aggregation, justifying its use as the default.

\subsection{Hyper-parameter Analysis}
\cref{fig:sensitivity} reports Recall@20 as we vary the base depth $L$, the curvature threshold $\tau$, and the dropout ratio. The optimum is consistent across all three datasets: $L=3$, $\tau=0.1$, and dropout $=0.2$. NTCF prefers a moderate $L$ (consistent with prior GCF observations~\cite{he2020lightgcn,wang2019neural}); deeper stacks ($L=4$) start to over-smooth, while a shallow $L=1$ removes the very signal NTCF is built to refine. A small but non-zero $\tau$ is preferred: $\tau\!=\!0.1$ keeps a non-trivial fraction of nodes in the curvature-adjusted regime, whereas $\tau\!=\!0.4$ collapses NTCF toward NGCF and degrades performance. Performance varies by less than $0.5$ percentage points (in absolute Recall@20) in a neighborhood of each optimum, indicating that NTCF is robust to hyper-parameter choice.

\begin{figure}[t]
\centering
\includegraphics[width=1\linewidth]{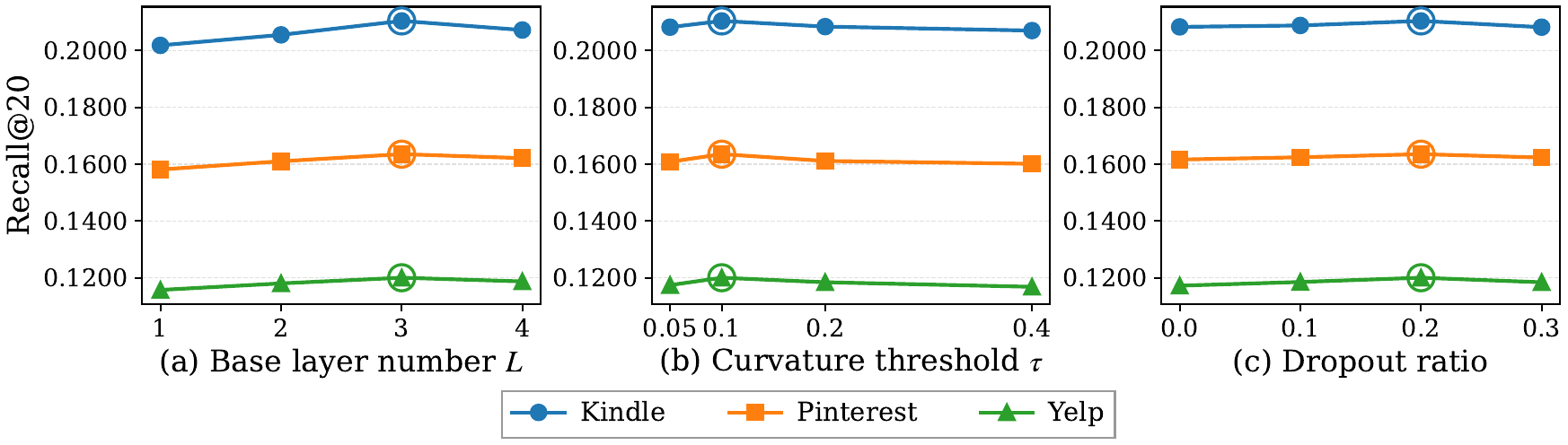}
\vskip -0.1in
\caption{Hyper-parameter analysis of NTCF across all datasets.}
\label{fig:sensitivity}
\vskip -0.15in
\end{figure}

\section{Conclusion}
In this paper, we revisit the structural assumptions of graph collaborative filtering and identify that fixing a global propagation depth ignores the heterogeneous local geometry of real user--item graphs. Motivated by this observation, we propose Neural Tree Collaborative Filtering (NTCF), which reinterprets each node's neighborhood as a rooted tree and assigns a node-specific propagation depth based on a closed-form local-degree-imbalance curvature score. We theoretically prove that NTCF strictly generalizes NGCF, degenerating to NGCF in the zero-adjustment limit (lower bound), and that, under the same Lipschitz contraction assumption used to analyze over-smoothing, NTCF maintains a non-vanishing Dirichlet-energy floor that uniform-depth NGCF provably loses as $L$ grows. Extensive experiments on three real-world datasets verify that NTCF consistently outperforms widely used GCF backbones and further boosts state-of-the-art self-supervised recommenders when used as their backbone.  We hope this motivates further work on geometry-aware, node-adaptive message-passing schedules for GCF.

\begin{acks}
This work was supported by the UGC General Research Fund no. 17209822 and the Innovation and Technology Commission Fund no. ITS/383/23FP from Hong Kong.
\end{acks}

\section*{GenAI Usage Disclosure}
The authors used generative AI tools (large language models) only for light language polishing of the manuscript text. All research ideas, theoretical results, model design, experimental implementation, data analysis, and figures were conceived, executed, and verified entirely by the authors. No generative AI was used to produce any code, experimental result, theorem proof, citation, or figure that appears in this work.

\balance
\bibliographystyle{ACM-Reference-Format}

\begin{thebibliography}{30}


\ifx \showCODEN    \undefined \def \showCODEN     #1{\unskip}     \fi
\ifx \showISBNx    \undefined \def \showISBNx     #1{\unskip}     \fi
\ifx \showISBNxiii \undefined \def \showISBNxiii  #1{\unskip}     \fi
\ifx \showISSN     \undefined \def \showISSN      #1{\unskip}     \fi
\ifx \showLCCN     \undefined \def \showLCCN      #1{\unskip}     \fi
\ifx \shownote     \undefined \def \shownote      #1{#1}          \fi
\ifx \showarticletitle \undefined \def \showarticletitle #1{#1}   \fi
\ifx \showURL      \undefined \def \showURL       {\relax}        \fi
\providecommand\bibfield[2]{#2}
\providecommand\bibinfo[2]{#2}
\providecommand\natexlab[1]{#1}
\providecommand\showeprint[2][]{arXiv:#2}

\bibitem[Cai et~al\mbox{.}(2023)]%
        {cailightgcl}
\bibfield{author}{\bibinfo{person}{Xuheng Cai}, \bibinfo{person}{Chao Huang}, \bibinfo{person}{Lianghao Xia}, {and} \bibinfo{person}{Xubin Ren}.} \bibinfo{year}{2023}\natexlab{}.
\newblock \showarticletitle{LightGCL: Simple Yet Effective Graph Contrastive Learning for Recommendation}. In \bibinfo{booktitle}{\emph{The Eleventh International Conference on Learning Representations}}.
\newblock


\bibitem[Chen et~al\mbox{.}(2020)]%
        {chen2020revisiting}
\bibfield{author}{\bibinfo{person}{Lei Chen}, \bibinfo{person}{Le Wu}, \bibinfo{person}{Richang Hong}, \bibinfo{person}{Kun Zhang}, {and} \bibinfo{person}{Meng Wang}.} \bibinfo{year}{2020}\natexlab{}.
\newblock \showarticletitle{Revisiting graph based collaborative filtering: A linear residual graph convolutional network approach}. In \bibinfo{booktitle}{\emph{Proceedings of the AAAI conference on artificial intelligence}}, Vol.~\bibinfo{volume}{34}. \bibinfo{pages}{27--34}.
\newblock


\bibitem[Chen et~al\mbox{.}(2025a)]%
        {chen2025hypercomplex}
\bibfield{author}{\bibinfo{person}{Zheyu Chen}, \bibinfo{person}{Jinfeng Xu}, \bibinfo{person}{Hewei Wang}, \bibinfo{person}{Shuo Yang}, \bibinfo{person}{Zitong Wan}, {and} \bibinfo{person}{Haibo Hu}.} \bibinfo{year}{2025}\natexlab{a}.
\newblock \showarticletitle{Hypercomplex Prompt-aware Multimodal Recommendation}. In \bibinfo{booktitle}{\emph{Proceedings of the 34th ACM International Conference on Information and Knowledge Management}}. \bibinfo{pages}{403--414}.
\newblock


\bibitem[Chen et~al\mbox{.}(2025b)]%
        {chen2025squeeze}
\bibfield{author}{\bibinfo{person}{Zheyu Chen}, \bibinfo{person}{Jinfeng Xu}, \bibinfo{person}{Yutong Wei}, {and} \bibinfo{person}{Ziyue Peng}.} \bibinfo{year}{2025}\natexlab{b}.
\newblock \showarticletitle{Squeeze and Excitation: A Weighted Graph Contrastive Learning for Collaborative Filtering}. In \bibinfo{booktitle}{\emph{Proceedings of the 48th International ACM SIGIR Conference on Research and Development in Information Retrieval}}. \bibinfo{pages}{2769--2773}.
\newblock


\bibitem[Geng et~al\mbox{.}(2015)]%
        {geng2015learning}
\bibfield{author}{\bibinfo{person}{Xue Geng}, \bibinfo{person}{Hanwang Zhang}, \bibinfo{person}{Jingwen Bian}, {and} \bibinfo{person}{Tat-Seng Chua}.} \bibinfo{year}{2015}\natexlab{}.
\newblock \showarticletitle{Learning image and user features for recommendation in social networks}. In \bibinfo{booktitle}{\emph{Proceedings of the IEEE international conference on computer vision}}. \bibinfo{pages}{4274--4282}.
\newblock


\bibitem[Glorot and Bengio(2010)]%
        {glorot2010understanding}
\bibfield{author}{\bibinfo{person}{Xavier Glorot} {and} \bibinfo{person}{Yoshua Bengio}.} \bibinfo{year}{2010}\natexlab{}.
\newblock \showarticletitle{Understanding the difficulty of training deep feedforward neural networks}. In \bibinfo{booktitle}{\emph{Proceedings of the thirteenth international conference on artificial intelligence and statistics}}. JMLR Workshop and Conference Proceedings, \bibinfo{pages}{249--256}.
\newblock


\bibitem[He and McAuley(2016)]%
        {he2016ups}
\bibfield{author}{\bibinfo{person}{Ruining He} {and} \bibinfo{person}{Julian McAuley}.} \bibinfo{year}{2016}\natexlab{}.
\newblock \showarticletitle{Ups and downs: Modeling the visual evolution of fashion trends with one-class collaborative filtering}. In \bibinfo{booktitle}{\emph{proceedings of the 25th international conference on world wide web}}. \bibinfo{pages}{507--517}.
\newblock


\bibitem[He et~al\mbox{.}(2020)]%
        {he2020lightgcn}
\bibfield{author}{\bibinfo{person}{Xiangnan He}, \bibinfo{person}{Kuan Deng}, \bibinfo{person}{Xiang Wang}, \bibinfo{person}{Yan Li}, \bibinfo{person}{Yongdong Zhang}, {and} \bibinfo{person}{Meng Wang}.} \bibinfo{year}{2020}\natexlab{}.
\newblock \showarticletitle{Lightgcn: Simplifying and powering graph convolution network for recommendation}. In \bibinfo{booktitle}{\emph{Proceedings of the 43rd International ACM SIGIR conference on research and development in Information Retrieval}}. \bibinfo{pages}{639--648}.
\newblock


\bibitem[He et~al\mbox{.}(2017)]%
        {he2017neural}
\bibfield{author}{\bibinfo{person}{Xiangnan He}, \bibinfo{person}{Lizi Liao}, \bibinfo{person}{Hanwang Zhang}, \bibinfo{person}{Liqiang Nie}, \bibinfo{person}{Xia Hu}, {and} \bibinfo{person}{Tat-Seng Chua}.} \bibinfo{year}{2017}\natexlab{}.
\newblock \showarticletitle{Neural collaborative filtering}. In \bibinfo{booktitle}{\emph{Proceedings of the 26th international conference on world wide web}}. \bibinfo{pages}{173--182}.
\newblock


\bibitem[Kingma and Ba(2014)]%
        {kingma2014adam}
\bibfield{author}{\bibinfo{person}{Diederik~P Kingma} {and} \bibinfo{person}{Jimmy Ba}.} \bibinfo{year}{2014}\natexlab{}.
\newblock \showarticletitle{Adam: A method for stochastic optimization}.
\newblock \bibinfo{journal}{\emph{arXiv preprint arXiv:1412.6980}} (\bibinfo{year}{2014}).
\newblock


\bibitem[Lin et~al\mbox{.}(2011)]%
        {lin2011ricci}
\bibfield{author}{\bibinfo{person}{Yong Lin}, \bibinfo{person}{Linyuan Lu}, {and} \bibinfo{person}{Shing-Tung Yau}.} \bibinfo{year}{2011}\natexlab{}.
\newblock \showarticletitle{Ricci curvature of graphs}.
\newblock \bibinfo{journal}{\emph{Tohoku Mathematical Journal}} \bibinfo{volume}{63}, \bibinfo{number}{4} (\bibinfo{year}{2011}), \bibinfo{pages}{605--627}.
\newblock


\bibitem[Lin et~al\mbox{.}(2022)]%
        {lin2022improving}
\bibfield{author}{\bibinfo{person}{Zihan Lin}, \bibinfo{person}{Changxin Tian}, \bibinfo{person}{Yupeng Hou}, {and} \bibinfo{person}{Wayne~Xin Zhao}.} \bibinfo{year}{2022}\natexlab{}.
\newblock \showarticletitle{Improving graph collaborative filtering with neighborhood-enriched contrastive learning}. In \bibinfo{booktitle}{\emph{Proceedings of the ACM web conference 2022}}. \bibinfo{pages}{2320--2329}.
\newblock


\bibitem[Liu et~al\mbox{.}(2021)]%
        {liu2021interest}
\bibfield{author}{\bibinfo{person}{Fan Liu}, \bibinfo{person}{Zhiyong Cheng}, \bibinfo{person}{Lei Zhu}, \bibinfo{person}{Zan Gao}, {and} \bibinfo{person}{Liqiang Nie}.} \bibinfo{year}{2021}\natexlab{}.
\newblock \showarticletitle{Interest-aware message-passing GCN for recommendation}. In \bibinfo{booktitle}{\emph{Proceedings of the web conference 2021}}. \bibinfo{pages}{1296--1305}.
\newblock


\bibitem[Liu et~al\mbox{.}(2020)]%
        {liu2020towards}
\bibfield{author}{\bibinfo{person}{Meng Liu}, \bibinfo{person}{Hongyang Gao}, {and} \bibinfo{person}{Shuiwang Ji}.} \bibinfo{year}{2020}\natexlab{}.
\newblock \showarticletitle{Towards deeper graph neural networks}. In \bibinfo{booktitle}{\emph{Proceedings of the 26th ACM SIGKDD international conference on knowledge discovery \& data mining}}. \bibinfo{pages}{338--348}.
\newblock


\bibitem[Mao et~al\mbox{.}(2021)]%
        {mao2021ultragcn}
\bibfield{author}{\bibinfo{person}{Kelong Mao}, \bibinfo{person}{Jieming Zhu}, \bibinfo{person}{Xi Xiao}, \bibinfo{person}{Biao Lu}, \bibinfo{person}{Zhaowei Wang}, {and} \bibinfo{person}{Xiuqiang He}.} \bibinfo{year}{2021}\natexlab{}.
\newblock \showarticletitle{UltraGCN: ultra simplification of graph convolutional networks for recommendation}. In \bibinfo{booktitle}{\emph{Proceedings of the 30th ACM international conference on information \& knowledge management}}. \bibinfo{pages}{1253--1262}.
\newblock


\bibitem[Ollivier(2009)]%
        {ollivier2009ricci}
\bibfield{author}{\bibinfo{person}{Yann Ollivier}.} \bibinfo{year}{2009}\natexlab{}.
\newblock \showarticletitle{Ricci curvature of {Markov} chains on metric spaces}.
\newblock \bibinfo{journal}{\emph{Journal of Functional Analysis}} \bibinfo{volume}{256}, \bibinfo{number}{3} (\bibinfo{year}{2009}), \bibinfo{pages}{810--864}.
\newblock


\bibitem[Oono and Suzuki(2020)]%
        {oono2020graph}
\bibfield{author}{\bibinfo{person}{Kenta Oono} {and} \bibinfo{person}{Taiji Suzuki}.} \bibinfo{year}{2020}\natexlab{}.
\newblock \showarticletitle{Graph Neural Networks Exponentially Lose Expressive Power for Node Classification}. In \bibinfo{booktitle}{\emph{International Conference on Learning Representations}}.
\newblock


\bibitem[Rendle et~al\mbox{.}(2009)]%
        {rendle2009bpr}
\bibfield{author}{\bibinfo{person}{Steffen Rendle}, \bibinfo{person}{Christoph Freudenthaler}, \bibinfo{person}{Zeno Gantner}, {and} \bibinfo{person}{Lars Schmidt-Thieme}.} \bibinfo{year}{2009}\natexlab{}.
\newblock \showarticletitle{BPR: Bayesian personalized ranking from implicit feedback}. In \bibinfo{booktitle}{\emph{Proceedings of the Twenty-Fifth Conference on Uncertainty in Artificial Intelligence}}. \bibinfo{pages}{452--461}.
\newblock


\bibitem[Wang et~al\mbox{.}(2019)]%
        {wang2019neural}
\bibfield{author}{\bibinfo{person}{Xiang Wang}, \bibinfo{person}{Xiangnan He}, \bibinfo{person}{Meng Wang}, \bibinfo{person}{Fuli Feng}, {and} \bibinfo{person}{Tat-Seng Chua}.} \bibinfo{year}{2019}\natexlab{}.
\newblock \showarticletitle{Neural graph collaborative filtering}. In \bibinfo{booktitle}{\emph{Proceedings of the 42nd international ACM SIGIR conference on Research and development in Information Retrieval}}. \bibinfo{pages}{165--174}.
\newblock


\bibitem[Xu et~al\mbox{.}(2024)]%
        {xu2024aligngroup}
\bibfield{author}{\bibinfo{person}{Jinfeng Xu}, \bibinfo{person}{Zheyu Chen}, \bibinfo{person}{Jinze Li}, \bibinfo{person}{Shuo Yang}, \bibinfo{person}{Hewei Wang}, {and} \bibinfo{person}{Edith~CH Ngai}.} \bibinfo{year}{2024}\natexlab{}.
\newblock \showarticletitle{AlignGroup: Learning and Aligning Group Consensus with Member Preferences for Group Recommendation}. In \bibinfo{booktitle}{\emph{Proceedings of the 33rd ACM International Conference on Information and Knowledge Management}}. \bibinfo{pages}{2682--2691}.
\newblock


\bibitem[Xu et~al\mbox{.}(2025a)]%
        {xu2024fourierkan}
\bibfield{author}{\bibinfo{person}{Jinfeng Xu}, \bibinfo{person}{Zheyu Chen}, \bibinfo{person}{Jinze Li}, \bibinfo{person}{Shuo Yang}, \bibinfo{person}{Wei Wang}, \bibinfo{person}{Xiping Hu}, {and} \bibinfo{person}{Edith Ngai}.} \bibinfo{year}{2025}\natexlab{a}.
\newblock \showarticletitle{Enhancing Graph Collaborative Filtering with FourierKAN Feature Transformation}. In \bibinfo{booktitle}{\emph{Proceedings of the 34th ACM International Conference on Information and Knowledge Management}}. \bibinfo{pages}{5376--5380}.
\newblock


\bibitem[Xu et~al\mbox{.}(2025b)]%
        {xu2025enhancing}
\bibfield{author}{\bibinfo{person}{Jinfeng Xu}, \bibinfo{person}{Zheyu Chen}, \bibinfo{person}{Jinze Li}, \bibinfo{person}{Shuo Yang}, \bibinfo{person}{Wei Wang}, \bibinfo{person}{Xiping Hu}, \bibinfo{person}{Raymond Chi-Wing Wong}, {and} \bibinfo{person}{Edith~CH Ngai}.} \bibinfo{year}{2025}\natexlab{b}.
\newblock \showarticletitle{Enhancing Robustness and Generalization Capability for Multimodal Recommender Systems via Sharpness-Aware Minimization}.
\newblock \bibinfo{journal}{\emph{IEEE Transactions on Knowledge and Data Engineering}} (\bibinfo{year}{2025}).
\newblock


\bibitem[Xu et~al\mbox{.}(2025c)]%
        {xu2025cohesion}
\bibfield{author}{\bibinfo{person}{Jinfeng Xu}, \bibinfo{person}{Zheyu Chen}, \bibinfo{person}{Wei Wang}, \bibinfo{person}{Xiping Hu}, \bibinfo{person}{Sang-Wook Kim}, {and} \bibinfo{person}{Edith~CH Ngai}.} \bibinfo{year}{2025}\natexlab{c}.
\newblock \showarticletitle{COHESION: Composite Graph Convolutional Network with Dual-Stage Fusion for Multimodal Recommendation}. In \bibinfo{booktitle}{\emph{Proceedings of the 48th International ACM SIGIR Conference on Research and Development in Information Retrieval}}. \bibinfo{pages}{1830--1839}.
\newblock


\bibitem[Xu et~al\mbox{.}(2025d)]%
        {xu2024mentor}
\bibfield{author}{\bibinfo{person}{Jinfeng Xu}, \bibinfo{person}{Zheyu Chen}, \bibinfo{person}{Shuo Yang}, \bibinfo{person}{Jinze Li}, \bibinfo{person}{Hewei Wang}, {and} \bibinfo{person}{Edith~CH Ngai}.} \bibinfo{year}{2025}\natexlab{d}.
\newblock \showarticletitle{Mentor: multi-level self-supervised learning for multimodal recommendation}. In \bibinfo{booktitle}{\emph{Proceedings of the AAAI Conference on Artificial Intelligence}}, Vol.~\bibinfo{volume}{39}. \bibinfo{pages}{12908--12917}.
\newblock


\bibitem[Xu et~al\mbox{.}(2025e)]%
        {xu2025nlgcl}
\bibfield{author}{\bibinfo{person}{Jinfeng Xu}, \bibinfo{person}{Zheyu Chen}, \bibinfo{person}{Shuo Yang}, \bibinfo{person}{Jinze Li}, \bibinfo{person}{Hewei Wang}, \bibinfo{person}{Wei Wang}, \bibinfo{person}{Xiping Hu}, {and} \bibinfo{person}{Edith Ngai}.} \bibinfo{year}{2025}\natexlab{e}.
\newblock \showarticletitle{NLGCL: Naturally Existing Neighbor Layers Graph Contrastive Learning for Recommendation}. In \bibinfo{booktitle}{\emph{Proceedings of the Nineteenth ACM Conference on Recommender Systems}}. \bibinfo{pages}{319--329}.
\newblock


\bibitem[Xu et~al\mbox{.}(2026b)]%
        {xu2026nlgclp}
\bibfield{author}{\bibinfo{person}{Jinfeng Xu}, \bibinfo{person}{Zheyu Chen}, \bibinfo{person}{Shuo Yang}, \bibinfo{person}{Jinze Li}, \bibinfo{person}{Hewei Wang}, \bibinfo{person}{Wei Wang}, \bibinfo{person}{Xiping Hu}, {and} \bibinfo{person}{Edith Ngai}.} \bibinfo{year}{2026}\natexlab{b}.
\newblock \showarticletitle{NLGCL+: Naturally Existing Neighbour Layers Graph Contrastive Learning with Adaptive Sample Weighting for Multimodal Recommendation}.
\newblock \bibinfo{journal}{\emph{ACM Transactions on Recommender Systems}} (\bibinfo{year}{2026}).
\newblock


\bibitem[Xu et~al\mbox{.}(2026a)]%
        {xu2026survey}
\bibfield{author}{\bibinfo{person}{Jinfeng Xu}, \bibinfo{person}{Zheyu Chen}, \bibinfo{person}{Shuo Yang}, \bibinfo{person}{Jinze Li}, \bibinfo{person}{Wei Wang}, \bibinfo{person}{Xiping Hu}, \bibinfo{person}{Steven Hoi}, {and} \bibinfo{person}{Edith Ngai}.} \bibinfo{year}{2026}\natexlab{a}.
\newblock \showarticletitle{A survey on multimodal recommender systems: Recent advances and future directions}.
\newblock \bibinfo{journal}{\emph{IEEE Transactions on Multimedia}} (\bibinfo{year}{2026}).
\newblock


\bibitem[Yu et~al\mbox{.}(2022)]%
        {yu2022graph}
\bibfield{author}{\bibinfo{person}{Junliang Yu}, \bibinfo{person}{Hongzhi Yin}, \bibinfo{person}{Xin Xia}, \bibinfo{person}{Tong Chen}, \bibinfo{person}{Lizhen Cui}, {and} \bibinfo{person}{Quoc Viet~Hung Nguyen}.} \bibinfo{year}{2022}\natexlab{}.
\newblock \showarticletitle{Are graph augmentations necessary? simple graph contrastive learning for recommendation}. In \bibinfo{booktitle}{\emph{Proceedings of the 45th international ACM SIGIR conference on research and development in information retrieval}}. \bibinfo{pages}{1294--1303}.
\newblock


\bibitem[Zhang et~al\mbox{.}(2024)]%
        {zhang2024recdcl}
\bibfield{author}{\bibinfo{person}{Dan Zhang}, \bibinfo{person}{Yangliao Geng}, \bibinfo{person}{Wenwen Gong}, \bibinfo{person}{Zhongang Qi}, \bibinfo{person}{Zhiyu Chen}, \bibinfo{person}{Xing Tang}, \bibinfo{person}{Ying Shan}, \bibinfo{person}{Yuxiao Dong}, {and} \bibinfo{person}{Jie Tang}.} \bibinfo{year}{2024}\natexlab{}.
\newblock \showarticletitle{RecDCL: Dual Contrastive Learning for Recommendation}. In \bibinfo{booktitle}{\emph{Proceedings of the ACM on Web Conference 2024}}. \bibinfo{pages}{3655--3666}.
\newblock


\bibitem[Zhao and Akoglu(2020)]%
        {zhao2020pairnorm}
\bibfield{author}{\bibinfo{person}{Lingxiao Zhao} {and} \bibinfo{person}{Leman Akoglu}.} \bibinfo{year}{2020}\natexlab{}.
\newblock \showarticletitle{PairNorm: Tackling Oversmoothing in GNNs}. In \bibinfo{booktitle}{\emph{International Conference on Learning Representations}}.
\newblock


\end{thebibliography}


\newpage

\end{document}